\documentclass[letterpaper,UKenglish,numberwithinsect]{lipics-v2016}
 
\usepackage{microtype}
\usepackage{algpseudocode,algorithm,epstopdf}
\usepackage{amsthm,amsmath,amssymb,amsfonts,cases,hyperref}
\usepackage{graphicx,subcaption}
\usepackage[capitalise,noabbrev]{cleveref}

\numberwithin{figure}{section}

\newcommand{\OPT}{\mbox{\sc OPT}}
\newcommand{\SOL}{\mbox{\sc SOL}}
\def\mcP{\mathcal{P}}

\title{A $(1.4 + \epsilon)$-approximation algorithm for the $2$-{\sc Max-Duo} problem%
\footnote{This work was partially supported by NSERC Canada and NSF China.}%
\footnote{An extended abstract appears in Proceedings of the 28th International Symposium on Algorithms and Computation (ISAAC 2017).
	LIPICS 92, Article No. 66, pp. 66:1--66:12.}}
\titlerunning{$(1.4 + \epsilon)$-approximation for $2$-{\sc Max-Duo}} 

\author[1]{Yao Xu}
\author[2,1]{Yong Chen}
\author[1]{Guohui Lin\footnote{Correspondence authors.}}
\author[3]{Tian Liu}
\author[4,1]{Taibo Luo}
\author[5]{Peng Zhang$^\ddagger$}
\affil[1]{Department of Computing Science, University of Alberta.  Edmonton, Alberta T6G 2E8, Canada.
  \texttt{\{xu2,taibo,guohui\}@ualberta.ca}}
\affil[2]{Department of Mathematics, Hangzhou Dianzi University.  Hangzhou, Zhejiang 310018, China.
  \texttt{chenyong@hdu.edu.cn}}
\affil[3]{Key Laboratory of High Confidence Software Technologies (MOE),
	Institute of Software, School of Electronic Engineering and Computer Science,
	Peking University. Beijing 100871, China.
  \texttt{lt@pku.edu.cn}}
\affil[4]{Business School, Sichuan University.  Chengdu, Sichuan 610065, China.}
\affil[5]{School of Computer Science and Technology, Shandong University.
  Jinan, Shandong 250101, China.
  \texttt{algzhang@sdu.edu.cn}}
\authorrunning{Y. Xu {\it et al.}} 

\Copyright{Yao Xu, Yong Chen, Guohui Lin, Tian Liu, Taibo Luo, Peng Zhang}


\begin{document}

\maketitle

\begin{abstract}
The {\em maximum duo-preservation string mapping} ({\sc Max-Duo}) problem is
the complement of the well studied {\em minimum common string partition} ({\sc MCSP}) problem,
both of which have applications in many fields including text compression and bioinformatics.
$k$-{\sc Max-Duo} is the restricted version of {\sc Max-Duo}, where every letter of the alphabet occurs at most $k$ times in each of the strings,
which is readily reduced into the well known {\em maximum independent set} ({\sc MIS}) problem on a graph of maximum degree $\Delta \le 6(k-1)$.
In particular, $2$-{\sc Max-Duo} can then be approximated arbitrarily close to $1.8$ using the state-of-the-art approximation algorithm for the {\sc MIS} problem.
$2$-{\sc Max-Duo} was proved APX-hard and very recently a $(1.6 + \epsilon)$-approximation was claimed, for any $\epsilon > 0$.
In this paper, we present a vertex-degree reduction technique,
based on which,
we show that $2$-{\sc Max-Duo} can be approximated arbitrarily close to $1.4$.
\end{abstract}

\section{Introduction}
The {\em minimum common string partition} ({\sc MCSP}) problem is a well-studied string comparison problem in computer science,
with applications in fields such as text compression and bioinformatics.
{\sc MCSP} was first introduced by Goldstein {\it et al.} \cite{GKZ04}, and can be defined as follows:
Consider two length-$n$ strings $A = (a_1, a_2, \ldots, a_n)$ and $B = (b_1, b_2, \ldots, b_n)$ over some alphabet $\Sigma$, such that $B$ is a permutation of $A$.
Let $\mcP_A$ be a {\em partition} of $A$, which is a multi-set of substrings whose concatenation in a certain order becomes $A$.
The {\em cardinality} of $\mcP_A$ is the number of substrings in $\mcP_A$.
The {\sc MCSP} problem asks to find a minimum cardinality partition $\mcP_A$ of $A$ which is also a partition of $B$.
$k$-{\sc MCSP} denotes the restricted version of {\sc MCSP} where every letter of the alphabet $\Sigma$ occurs at most $k$ times in each of the two strings.

Goldstein {\it et al.} \cite{GKZ04} have shown that the {\sc MCSP} problem is NP-hard and APX-hard, even when $k = 2$.
There have been several approximation algorithms \cite{CZF05,CKS04,CM07,GKZ04,KW06,KW07} presented since 2004, 
among which the current best result is an $O(\log n \log^* n)$-approximation for the general {\sc MCSP} and an $O(k)$-approximation for $k$-{\sc MCSP}.
On the other hand, {\sc MCSP} is proved to be {\em fixed parameter tractable} (FPT), with respect to $k$ and/or the cardinality of the optimal partition \cite{Dam08,JZZ12,BFK13,BK14}.

An ordered pair of consecutive letters in a string is called a {\em duo} of the string \cite{GKZ04},
which is said to be {\em preserved} by a partition if the pair resides inside a substring of the partition.
Therefore, a length-$\ell$ substring in the partition {\em preserves} $\ell-1$ duos of the string.
With the complementary objective to that of {\sc MCSP},
the problem of maximizing the number of duos preserved in the common partition is referred to as
the {\em maximum duo-preservation string mapping} problem by Chen {\it et al.} \cite{CCS14}, denoted as {\sc Max-Duo}.
Analogously, $k$-{\sc Max-Duo} is the restricted version of {\sc Max-Duo}
where every letter of the alphabet $\Sigma$ occurs at most $k$ times in each string.
In this paper, we focus on $2$-{\sc Max-Duo}, to design an improved approximation algorithm.

Along with {\sc Max-Duo}, Chen {\it et al.} \cite{CCS14} introduced the {\em constrained maximum induced subgraph} ({\sc CMIS}) problem,
in which one is given an $m$-partite graph $G = (V_1, V_2, \ldots, V_m, E)$ with each $V_i$ having $n_i^2$ vertices arranged in an $n_i \times n_i$ matrix,
and the goal is to find $n_i$ vertices in each $V_i$ from different rows and different columns such that
the number of edges in the induced subgraph is maximized.  
$k$-{\sc CMIS} is the restricted version of {\sc CMIS} where $n_i \le k$ for all $i$.
Given an instance of {\sc Max-Duo},
we may construct an instance of {\sc CMIS} by setting $m$ to be the number of distinct letters in the string $A$,
and $n_i$ to be the number of occurrences of the $i$-th distinct letter;
the vertex in the $(s, t)$-entry of the $n_i \times n_i$ matrix ``means''
mapping the $s$-th occurrence of the $i$-th distinct letter in the string $A$ to its $t$-th occurrence in the string $B$;
and there is an edge between a vertex of $V_i$ and a vertex of $V_j$ if
the two corresponding mappings together preserve a duo.
Therefore, {\sc Max-Duo} is a special case of {\sc CMIS}, and furthermore $k$-{\sc Max-Duo} is a special case of $k$-{\sc CMIS}.
Chen {\it et al.} \cite{CCS14} presented a $k^2$-approximation for $k$-{\sc CMIS} and a $2$-approximation for $2$-{\sc CMIS},
based on a linear programming and randomized rounding techniques.
These imply that $k$-{\sc Max-Duo} can also be approximated within a ratio of $k^2$ and $2$-{\sc Max-Duo} can be approximated within a ratio of $2$.

Alternatively, an instance of the $k$-{\sc Max-Duo} problem with the two strings $A = (a_1, a_2, \ldots, a_n)$ and $B = (b_1, b_2, \ldots, b_n)$
can be viewed as a bipartite graph $H = (A, B, F)$, constructed as follows:
The vertices in $A$ and $B$ are $a_1, a_2, \ldots, a_n$ in order and $b_1, b_2, \ldots, b_n$ in order, respectively,
and there is an edge between $a_i$ and $b_j$ if they are the same letter.
The two edges $(a_i, b_j), (a_{i+1}, b_{j+1}) \in F$ are called a pair of {\em parallel} edges.
This way, a common partition of the strings $A$ and $B$ corresponds one-to-one to a perfect matching in $H$,
and the number of duos preserved by the partition is exactly the number of pairs of parallel edges in the matching.

Moreover, from the bipartite graph $H = (A, B, F)$,
we can construct another graph $G = (V, E)$ in which every vertex of $V$ corresponds to a pair of parallel edges of $F$,
and there is an edge between two vertices of $V$ if the two corresponding pairs of parallel edges of $F$ {\em cannot} co-exist in any perfect matching of $H$ 
(called {\em conflicting}, which can be determined in constant time; see Section \ref{sec2} for more details).
This way, one easily sees that a set of duos that can be preserved together, by a perfect matching of $H$,
corresponds one-to-one to an independent set of $G$ \cite{GKZ04,BKL14}.
Therefore, the {\sc Max-Duo} problem can be cast as a special case of the well-known {\em maximum independent set} ({\sc MIS}) problem~\cite{GJ79};
furthermore, Boria {\it et al.} \cite{BKL14} showed that in such a reduction,
an instance of $k$-{\sc Max-Duo} gives rise to a graph with a maximum degree $\Delta \le 6(k-1)$.
It follows that the state-of-the-art $\big((\Delta + 3)/{5} + \epsilon\big)$-approximation algorithm for {\sc MIS} \cite{BF99}, for any $\epsilon > 0$,
is a ${\big((6k-3)}/{5} + \epsilon\big)$-approximation algorithm for $k$-{\sc Max-Duo}.
Especially, $2$-{\sc Max-Duo} can now be better approximated within a ratio of $1.8 + \epsilon$.
Boria {\it et al.} \cite{BKL14} proved that $2$-{\sc Max-Duo} is APX-hard, similar to $2$-{\sc MCSP} \cite{GKZ04},
via a linear reduction from {\sc MIS} on cubic graphs.
For {\sc MIS} on cubic graphs, it is NP-hard to approximate within $1.00719$ \cite{BK99}.
Besides, Boria {\it et al.} \cite{BKL14} claimed that $2$-{\sc Max-Duo} can be approximated within $1.6 + \epsilon$, for any $\epsilon > 0$.

Recently, Boria {\it et al.} \cite{BCC16} presented a local search $3.5$-approximation for the general {\sc Max-Duo} problem.
In the meantime, Brubach \cite{Bru16} presented a $3.25$-approximation using a novel {\em combinatorial triplet matching}.
{\sc Max-Duo} has also been proved to be FPT by Beretta {\it et al.} \cite{BCD16}, with respect to the number of preserved duos in the optimal partition.
Most recently, two local search algorithms were independently designed for the general {\sc Max-Duo} problem at the same time,
achieving approximation ratios of $2.917$ \cite{XCLL17a} and $2 + \epsilon$ \cite{DGO17} for any $\epsilon > 0$, respectively.
They both exceed the previously the best ${\big((6k-3)}/{5} + \epsilon\big)$-approximation algorithm for $k$-{\sc Max-Duo}, when $k \ge 3$.
In this paper, we focus on the $2$-{\sc Max-Duo} problem;
using the above reduction to the {\sc MIS} problem, we present a vertex-degree reduction scheme and design an improved $(1.4 + \epsilon)$-approximation,
for any $\epsilon > 0$.

The rest of the paper is organized as follows.
We provide some preliminaries in Section \ref{sec2}, including several important structural properties of the graph constructed from the two given strings.
The vertex-degree reduction scheme is also presented as a separate subsection in Section \ref{sec2}.
The new approximation algorithm, denoted as {\sc Approx}, is presented in Section \ref{sec3},
where we show that it is a $(1.4 + \epsilon)$-approximation for $2$-{\sc Max-Duo}.
We conclude the paper in Section \ref{sec5}.

\section{Preliminaries}
\label{sec2}
Consider an instance of the $k$-{\sc Max-Duo} problem with two length-$n$ strings $A = (a_1, a_2, \ldots, a_n)$ and $B = (b_1, b_2, \ldots, b_n)$
such that $B$ is a permutation of $A$.
Recall that we can view the instance as a bipartite graph $H = (A, B, F)$,
where the vertices in $A$ and $B$ are $a_1, a_2, \ldots, a_n$ in order and $b_1, b_2, \ldots, b_n$ in order, respectively,
and there is an edge between $a_i \in A$ and $b_j \in B$ if they are the same letter, denoted as $e_{i, j}$.
See Figure~\ref{fig21a} for an example, where $A = (a, b, c, d, e, f, b, c, d, e)$ and $B = (f, b, c, d, e, a, b, c, d, e)$.
Note that $|F| \le kn$, and so $H$ can be constructed in $O(n^2)$ time.

\begin{figure}[htb]
\centering
\begin{subfigure}{0.45\textwidth}
\includegraphics[width=\linewidth]{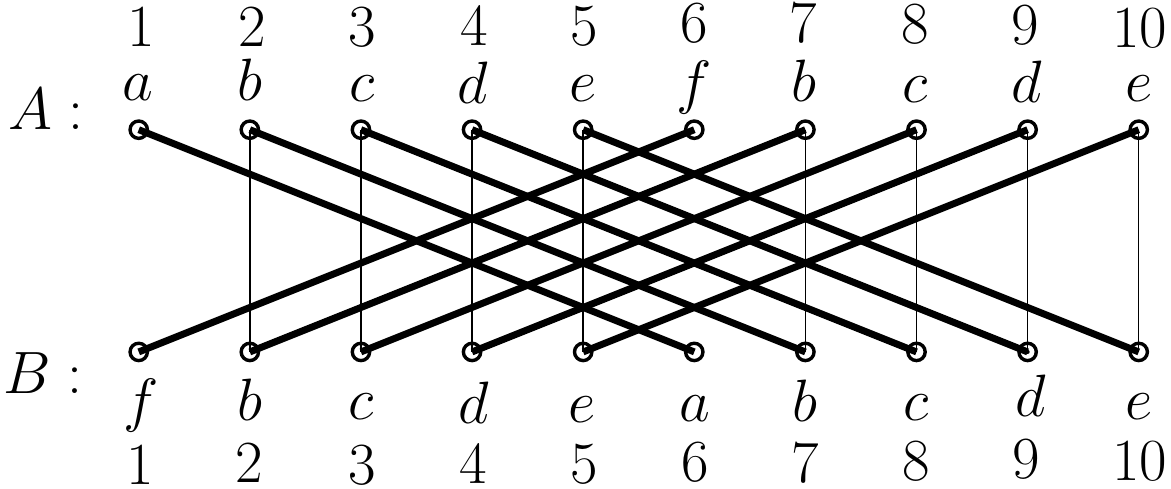}
\caption{The bipartite graph $H = (A, B, F)$, where the ten edges in bold form a perfect matching.\label{fig21a}}
\end{subfigure}

\vskip 0.2in
\begin{subfigure}{0.6\textwidth}
\includegraphics[width=\linewidth]{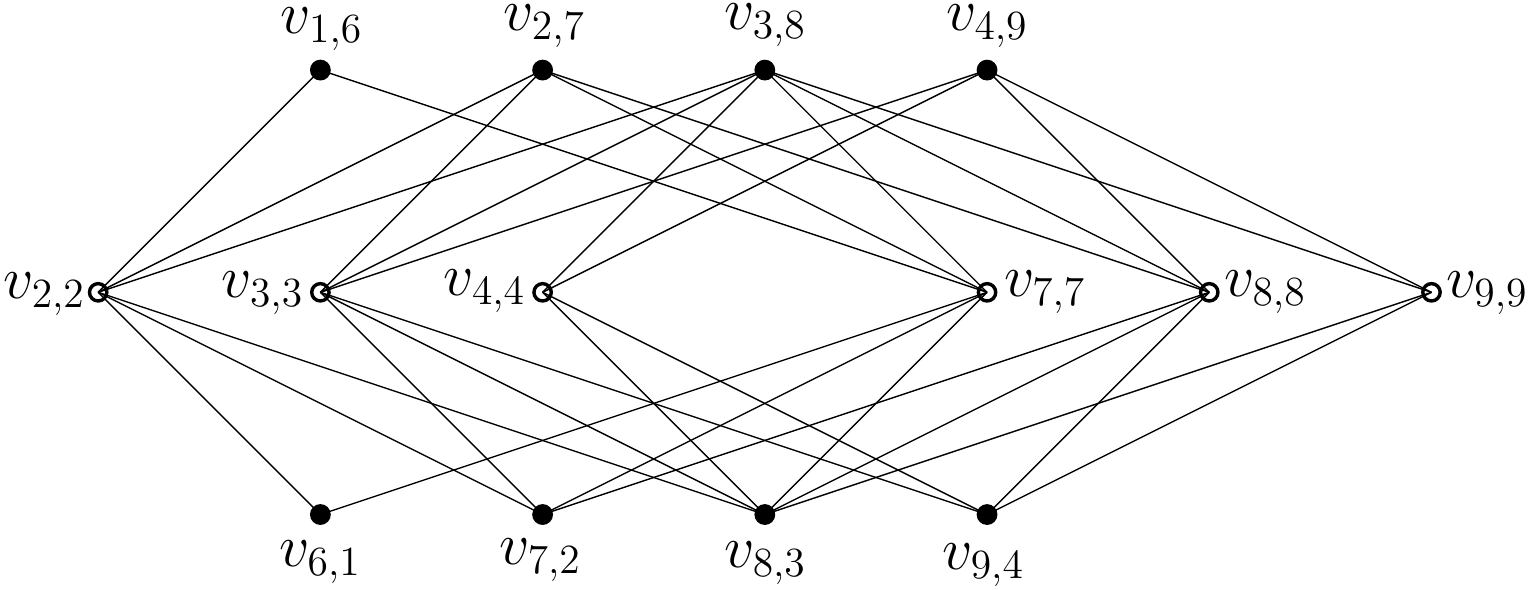}
\caption{The instance graph $G = (V, E)$ of {\sc MIS}, where the eight filled vertices form an independent set.\label{fig21b}}
\end{subfigure}
\caption{An instance of the $k$-{\sc Max-Duo} problem with $A = (a, b, c, d, e, f, b, c, d, e)$ and $B = (f, b, c, d, e, a, b, c, d, e)$.
	\cref{fig21a} is the graphical view as a bipartite graph $H = (A, B, F)$,
	where a perfect matching consisting of the ten bold edges form into eight pairs of parallel edges,
	corresponding to the eight preserved duos $(a, b), (b, c), (c, d), (d, e), (f, b), (b, c), (c, d)$ and $(d, e)$.
	\cref{fig21b} shows the instance graph $G = (V, E)$ of {\sc MIS} constructed from $H$,
	where the independent set $\{v_{1, 6}, v_{2, 7}, v_{3, 8}, v_{4, 9}, v_{6, 1}, v_{7, 2}, v_{8, 3}, v_{9, 4}\}$
	corresponds to the eight pairs of parallel edges shown in \cref{fig21a}, and consequently also corresponds to the eight preserved duos.
	In this instance, we have $k = 2$.
	Any maximum independent set of $G$ must contain some of the degree-$6$ vertices,
	invalidating the $(1.6 + \epsilon)$-approximation for $2$-{\sc Max-Duo} proposed in \cite{BKL14}.\label{fig21}}
\end{figure}

The two edges $e_{i, j}, e_{i+1, j+1} \in F$ are called a pair of {\em parallel} edges (and they are said to be parallel to each other);
when both are included in a perfect matching of $H$, the corresponding duo $(a_i, a_{i+1})$ of $A$ is preserved.
Two pairs of parallel edges are {\em conflicting} if they cannot co-exist in any perfect matching of $H$.
This motivates the following reduction from the $k$-{\sc Max-Duo} problem to the {\sc MIS} problem:
From the bipartite graph $H = (A, B, F)$, we construct another graph $G = (V, E)$ in which
a vertex $v_{i, j}$ of $V$ corresponds to the pair of parallel edges $(e_{i, j}, e_{i+1, j+1})$ of $F$;
two vertices of $V$ are {\em conflicting} if and only if the two corresponding pairs of parallel edges are conflicting,
and two conflicting vertices of $V$ are adjacent in $G$.
One can see that a set of duos of $A$ that can be preserved all together,
a set of pairwise non-conflicting pairs of parallel edges of $F$,
and an independent set in $G$,
are equivalent to each other.
See Figure \ref{fig21b} for an example of the graph $G = (V, E)$ constructed from the bipartite graph $H$ shown in Figure \ref{fig21a}.
We note that $|V| \le k (n - 1)$ and thus $G$ can be constructed in $O(k^2 n^2)$ time from the instance of the $k$-{\sc Max-Duo} problem.

In the graph $G$, for any $v \in V$, we use $N(v)$ to denote the set of its neighbors, that is, the vertices adjacent to $v$.
The two ordered letters in the duo corresponding to the vertex $v$ is referred to as the {\em letter content} of $v$.
For example, in Figure \ref{fig21b}, the letter content of $v_{1,6}$ is ``$ab$'' and the letter content of $v_{6,1}$ is ``$fb$''.

Recall from the construction that there is an edge $e_{i,j}$ in the graph $H = (A, B, F)$ if $a_i = b_j$,
and there is a vertex $v_{i,j}$ in the graph $G = (V, E)$ if the parallel edges $e_{i,j}$ and $e_{i+1,j+1}$ are in $H = (A, B, F)$.

\begin{lemma}
\label{lemma21}
The graph $G = (V, E)$ has the following properties.
\begin{enumerate}
\item
	If $v_{i, j}$, $v_{i+2, j+2} \in V$, then $v_{i+1, j+1} \in V$.
\item
	Given any subset of vertices $V' \subset V$, 
	let $F' = \{e_{i, j} | v_{i, j} \in V'\}$,
	$A' = \{a_i | e_{i, j} \in F'\}$, 
	and $B' = \{b_j | e_{i, j} \in F'\}$.
	If the subgraph $H' = (A', B', F')$ in $H$ is connected,
	then all the vertices of $V'$ have the same letter content;
	and consequently for any two vertices $v_{i, j}, v_{h, \ell} \in V'$, we have both $v_{h, j}, v_{i, \ell} \in V$.
\item
	For any $v_{i, j} \in V$, we have
	\begin{equation}
	\label{eq1}
	N(v_{i, j}) = 
	\bigcup_{p = -1, 0, 1} \{v_{i'+p, j+p} \in V \mid i' \ne i\} \cup
	\bigcup_{p = -1, 0, 1} \{v_{i+p, j'+p} \in V \mid j' \ne j\}.
	\end{equation}
\end{enumerate}
\end{lemma}
\begin{proof}
By definition, $v_{i, j} \in V$ if and only if $e_{i, j}, e_{i+1, j+1} \in F$.
\begin{enumerate}
\item
	If also $v_{i+2, j+2} \in V$, that is, $e_{i+2, j+2}$, $e_{i+3, j+3} \in F$,
	then $e_{i+1, j+1}, e_{i+2, j+2} \in F$ leading to $v_{i+1, j+1} \in V$.

\item 
	Note that an edge $e_{i, j} \in F$ if and only if the two vertices $a_i$ and $b_j$ are the same letter, 
	and clearly each connected component in $H$ is complete bipartite and all the vertices are the same letter.
	It follows that if the induced subgraph $H' = (A', B', F')$ in $H$ is connected,
	then all its vertices are the same letter;
	furthermore, all the duos starting with these vertices have the same letter content;
	and therefore for any two vertices $v_{i, j}, v_{h, \ell} \in V'$, both $v_{h, j}, v_{i, \ell} \in V$.

\item 
	For any vertex $v_{i,j}$, or equivalently the pair of parallel edges $(e_{i, j}, e_{i+1, j+1})$ in $F$,
	which are incident at four vertices $a_i, a_{i+1}, b_j, b_{j+1}$,
	a conflicting pair of parallel edges can be one of the six kinds:
	to share exactly one of the four vertices $a_i, a_{i+1}, b_j, b_{j+1}$,
	to share exactly two vertices $a_i$ and $a_{i+1}$, and
	to share exactly two vertices $b_j$ and $b_{j+1}$.
	The sets of these six kinds of conflicting pairs are as described in the lemma, for example, 
	$\{v_{i'-1, j-1} \in V \mid i' \ne i\}$ is the set of conflicting pairs each sharing only the vertex $b_j$ with the pair $v_{i, j}$.
\end{enumerate}
\end{proof}

From Lemma~\ref{lemma21} and its proof, we see that for any vertex of $V$ there are at most $k-1$ conflicting vertices of each kind
(corresponding to a set in \cref{eq1}).
We thus have the following corollary.

\begin{corollary}
\label{coro22}
The maximum degree of the vertices in $G = (V, E)$ is $\Delta \le 6(k-1)$.
\end{corollary}

\subsection{When $k = 2$}
We examine more properties for the graph $G = (V, E)$ when $k = 2$.
First, from Corollary~\ref{coro22} we have $\Delta \le 6$.

Berman and Fujito \cite{BF99} have presented an approximation algorithm with a performance ratio arbitrarily close to ${(\Delta+3)}/{5}$ for the {\sc MIS} problem,
on graphs with maximum degree $\Delta$.
This immediately implies a $(1.8 + \epsilon)$-approximation for $2$-{\sc Max-Duo}.
Our goal is to reduce the maximum degree of the graph $G = (V, E)$ to achieve a better approximation algorithm.
To this purpose, we examine all the degree-$6$ and degree-$5$ vertices in the graph $G$,
and show a scheme to safely remove them from consideration when computing an independent set.
This gives rise to a new graph $G_2$ with maximum degree at most $4$, leading to a desired $(1.4 + \epsilon)$-approximation for $2$-{\sc Max-Duo}.

We remark that, in our scheme we first remove the degree-$6$ vertices from $G$ to compute an independent set,
and later we add half of these degree-$6$ vertices to the computed independent set to become the final solution.
Contrary to the claim that there always exists a maximum independent set in $G$ containing no degree-$6$ vertices~\cite[Lemma 1]{BKL14},
the instance in \cref{fig21} shows that any maximum independent set for the instance must contain some degree-$6$ vertices,
thus invalidating the $(1.6 + \epsilon)$-approximation for $2$-{\sc Max-Duo} proposed in \cite{BKL14}.

In more details, the instance of $2$-{\sc Max-Duo}, illustrated in \cref{fig21},
consists of two length-$10$ strings $A = (a, b, c, d, e, f, b, c, d, e)$ and $B = (f, b, c, d, e, a, b, c, d, e)$.
The bipartite graph $H = (A, B, F)$ is shown in \cref{fig21a} and the instance graph $G = (V, E)$ of the {\sc MIS} problem is shown in \cref{fig21b}.
In the graph $G$, we have six degree-$6$ vertices: $v_{2, 2}, v_{7, 7}, v_{3, 3}, v_{3, 8}, v_{8, 3}$ and $v_{8, 8}$.
One can check that $\{v_{1, 6}, v_{2, 7}$, $v_{3, 8}, v_{4, 9}$, $v_{6, 1}, v_{7, 2}$, $v_{8, 3}, v_{9, 4}\}$ is an independent set in $G$, of size $8$.
On the other hand, if none of these degree-$6$ vertices is included in an independent set,
then because the four vertices $v_{4, 4}, v_{4, 9}, v_{9, 4}, v_{9, 9}$ form a square implying that at most two of them can be included in the independent set,
the independent set would be of size at most $6$, and thus can never be a maximum independent set in $G$.

Consider a duo $(a_i, a_{i+1})$ of the string $A$ and for ease of presentation assume its letter content is ``$ab$''.
If no duo of the string $B$ has the same letter content ``$ab$'',
then this duo of the string $A$ can never be preserved;
in fact this duo does not even become (a part of) a vertex of $V$ of the graph $G$.
If there is exactly one duo $(b_j, b_{j+1})$ of the string $B$ having the same letter content ``$ab$'',
then these two duos make up a vertex $v_{i, j} \in V$,
and from Lemma~\ref{lemma21} we know that the degree of the vertex $v_{i, j} \in V$ is at most $5$,
since there is no such vertex $v_{i, j'}$ with $j' \ne j$ sharing exactly the two letters $a_i$ and $a_{i+1}$ with $v_{i, j}$.
Therefore, if the degree of the vertex $v_{i, j} \in V$ is six,
then there must be two duos of the string $A$ and two duos of the string $B$ having the same letter content ``$ab$''.
Assume the other duo of the string $A$ and the other duo of the string $B$ having the same letter content ``$ab$''
are $(a_{i'}, a_{i'+1})$ and $(b_{j'}, b_{j'+1})$, respectively.
Then all four vertices $v_{i, j}, v_{i, j'}, v_{i', j}, v_{i', j'}$ exist in $V$.
We call the subgraph of $G$ induced on these four vertices a {\em square},
and denote it as $S(i,i';j,j') = (V(i,i';j,j'), E(i,i';j,j'))$,
where $V(i,i';j,j') = \{v_{i, j}, v_{i, j'}, v_{i', j}, v_{i', j'}\}$ and
$E(i,i';j,j') = \{(v_{i, j}, v_{i, j'}), (v_{i, j}, v_{i', j}), (v_{i', j'}, v_{i, j'}), (v_{i', j'}, v_{i', j})\}$ due to their conflicting relationships.
One clearly sees that every square has a unique letter content, which is the letter content of its four member vertices.

In \cref{fig21b}, there are three squares $S(2,7; 2,7)$, $S(3,8; 3, 8)$ and $S(4,9; 4,9)$,
with their letter contents ``$bc$'', ``$cd$'' and ``$de$'', respectively.
The above argument says that every degree-$6$ vertex of $V$ must belong to a square, but the converse is not necessarily true,
for example, all vertices of the square $S(4,9; 4,9)$ have degree $4$.
We next characterize several properties of a square.

The following lemma is a direct consequence of how the graph $G$ is constructed and the fact that $k = 2$.

\begin{lemma}
\label{lemma23}
In the graph $G = (V, E)$ constructed from an instance of $2$-{\sc Max-Duo},
\begin{enumerate}
\item 
	for each index $i$, there are at most two distinct $j$ and $j'$ such that $v_{i, j}, v_{i, j'} \in V$;
\item 
	if $v_{i, j}, v_{i, j'} \in V$ where $j' \ne j$, and $v_{i+1, j''+1} \in V$ (or symmetrically, $v_{i-1, j''-1} \in V$),
	then either $j'' = j$ or $j'' = j'$.
\end{enumerate}
\end{lemma}

\begin{lemma}
\label{lemma24}
For any square $S(i,i'; j,j')$ in the graph $G = (V, E)$,
$N(v_{i, j}) = N(v_{i', j'})$, $N(v_{i, j'}) = N(v_{i', j})$,
and $N(v_{i, j}) \cap N(v_{i, j'}) = \emptyset$.
(Together, these imply that every vertex of $V$ is adjacent to either none or exactly two of the four member vertices of a square.)
\end{lemma}
\begin{proof}
Consider the two vertices $v_{i, j}$ and $v_{i', j'}$, which have common neighbors $v_{i, j'}$ and $v_{i', j}$ in the square.

Note that $v_{i, j'}$ and $v_{i, j}$ share both the letters $a_i$ and $a_{i+1}$.
If there is a vertex adjacent to $v_{i, j}$ by sharing $a_{i+1}$ but not $a_i$, then this vertex is $v_{i+1, j''+1}$ with $j'' \ne j$,
and thus it has to be $v_{i+1, j'+1}$ (by Lemma~\ref{lemma23}).
We consider two subcases:
If $i+1 = i'-1$, then $j'+1 = j-1$ due to $k = 2$.
Thus, this vertex $v_{i+1, j'+1}$ actually shares $a_{i+1}$ and $b_j$ with $v_{i, j}$;
also, it shares $a_{i'}$ and $b_{j'+1}$ with $v_{i', j'}$;
and therefore it is adjacent to $v_{i', j'}$ too, but not adjacent to $v_{i, j'}$ or $v_{i', j}$.
If $i+1 \ne i'-1$, then this vertex $v_{i+1, j'+1}$ shares only $a_{j+1}$ with the vertex $v_{i, j}$;
also it shares only $b_{j'+1}$ with $v_{i', j'}$;
and therefore it is adjacent to $v_{i', j'}$ too, but not adjacent to $v_{i, j'}$ or $v_{i', j}$.

The other three symmetric cases can be discussed exactly the same and the lemma is proved.
\end{proof}

\begin{corollary}
\label{coro25}
In the graph $G = (V, E)$,
the degree-$6$ vertices can be partitioned into pairs, 
where each pair of degree-$6$ vertices belong to a square in $G$ and they are adjacent to the same six other vertices,
two inside the square and four outside of the square.
\end{corollary}
\begin{proof}
We have seen that every degree-$6$ vertex in the graph $G$ must be in a square.
The above Lemma~\ref{lemma24} states that the four vertices of a square $S(i,i'; j,j')$ can be partitioned into two pairs,
$\{v_{i, j}, v_{i', j'}\}$ and $\{v_{i, j'}, v_{i', j}\}$,
and the two vertices inside each pair are non-adjacent to each other and have the same neighbors.
In particular, if the vertex $v_{i, j}$ in the square $S(i,i'; j,j')$ has degree $6$, then Lemma~\ref{lemma21} states that it is adjacent to the six vertices
$v_{i-1, j'-1}, v_{i, j'}, v_{i+1, j'+1}, v_{i'-1, j-1}, v_{i', j}, v_{i'+1, j+1}$ (see an illustration in \cref{fig23}).
\end{proof}

\begin{figure}[H]
\centering
\includegraphics[width=0.62\linewidth]{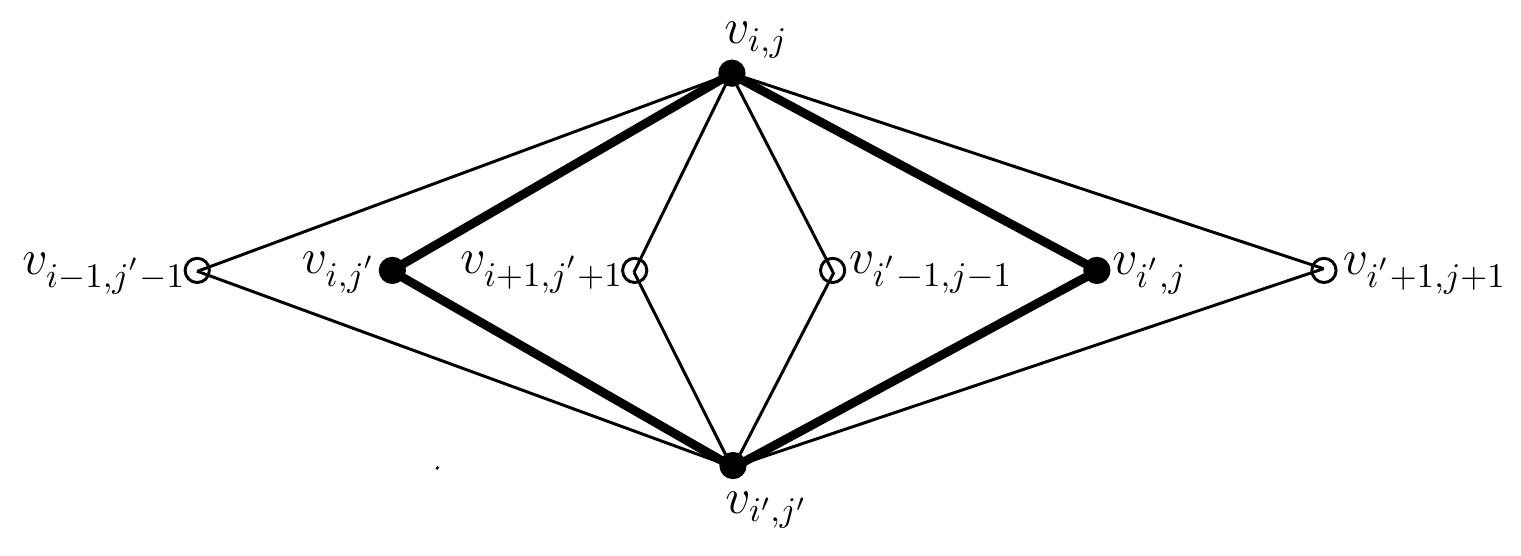}
\caption{The square $S(i,i'; j,j')$ shown in bold lines.
	The two non-adjacent vertices $v_{i, j}$ and $v_{i', j'}$ of the square form a pair stated in Corollary~\ref{coro25};
	they have $6$ common neighbors, of which two inside the square and four outside of the square.\label{fig23}}
\end{figure}

\begin{corollary}
\label{coro26}
If there is no square in the graph $G = (V, E)$, 
then every degree-$5$ vertex is adjacent to a degree-$1$ vertex.
\end{corollary}
\begin{proof}
Assume the vertex $v_{i, j}$ has degree $5$.
Due to the non-existence of any square in the graph $G$ and Lemma~\ref{lemma21},
either there is no vertex sharing exactly the two letters $a_i$ and $a_{i+1}$ with $v_{i, j}$,
or there is no vertex sharing exactly the two letters $b_j$ and $b_{j+1}$ with $v_{i, j}$.
We assume without loss of generality that there is no vertex sharing exactly the two letters $a_i$ and $a_{i+1}$ with $v_{i, j}$,
and furthermore assume $v_{i', j}$, $i' \ne i$, is the vertex sharing exactly the two letters $b_j$ and $b_{j+1}$ with $v_{i, j}$.

It follows that $N(v_{i, j}) = \{v_{i-1, j''-1}, v_{i+1, j'''+1}, v_{i'-1, j-1}, v_{i', j}, v_{i'+1, j+1}\}$, for some $j'' \ne j$ and $j''' \ne j$.
Due to $k = 2$, this implies that $a_{i-1} \ne b_{j-1} = a_{i'-1}$ and $a_{i+2} \ne b_{j+2} = a_{i'+2}$.
Therefore, there is no vertex of $V$ sharing exactly the letter $a_{i'}$ ($a_{i'+1}, b_{j}, b_{j+1}$, respectively) with the vertex $v_{i', j}$,
neither a vertex of $V$ sharing exactly the two letters $a_{i'}$ and $a_{i'+1}$ with the vertex $v_{i', j}$.
That is, the vertex $v_{i', j}$ is adjacent to only $v_{i, j}$ in the graph $G$. 
\end{proof}

We say the two vertices $v_{i, j}$ and $v_{i+1, j+1}$ of $V$ are {\em consecutive};
and we say the two squares $S(i,i';j,j')$ and $S(i+1,i'+1;j+1,j'+1)$ in $G$ are {\em consecutive}.
Clearly, two consecutive squares contain four pairs of consecutive vertices.
The following Lemma~\ref{lemma27} summarizes the fact that when two consecutive vertices belong to two different squares,
then these two squares are also consecutive (and thus contain the other three pairs of consecutive vertices).

\begin{lemma}
\label{lemma27}
In the graph $G$,
if there are two consecutive vertices $v_{i, j}$ and $v_{i+1, j+1}$ belonging to
two different squares $S(i_1,i'_1; j_1,j'_1)$ and $S(i_2,i'_2; j_2,j'_2)$ respectively,
then $i_2 = i_1+1, i'_2 = i'_1+1, j_2 = j_1+1, j'_2 = j'_1+1$, i.e., these two squares are consecutive.
\end{lemma}
\begin{proof}
This is a direct result of the fact that no two distinct squares have any member vertex in common.
\end{proof}

A series of $p$ consecutive squares $\{S(i+q, i'+q; j+q, j'+q), q = 0, 1, \ldots, p-1\}$ in the graph $G$, where $p \ge 1$,
is {\em maximal} if none of the square $S(i-1, i'-1; j-1, j'-1)$ and the square $S(i+p, i'+p; j+p, j'+p)$ exists in the graph $G$.
Note that the non-existence of the square $S(i-1, i'-1; j-1, j'-1)$ in $G$ does not rule out the existence of some of the four vertices
$v_{i-1, j-1}, v_{i'-1, j'-1}, v_{i-1, j'-1}, v_{i'-1, j-1}$ in $V$;
in fact by Lemma~\ref{lemma21} there can be as many as two of these four vertices existing in $V$
(however, more than two would imply the existence of the square).
Similarly, there can be as many as two of the four vertices
$v_{i+p, j+p}, v_{i'+p, j'+p}, v_{i+p, j'+p}, v_{i'+p, j+p}$ existing in $V$.
In the sequel, a maximal series of $p$ consecutive squares starting with $S(i, i'; j, j')$ is denoted as ${\cal S}^p(i, i'; j, j')$, where $p \ge 1$.
See for an example in \cref{fig24b} where there is a maximal series of $2$ consecutive squares ${\cal S}^2(2, 8; 2, 8)$,
where the instance of the $2$-{\sc Max-Duo} is expanded slightly from the instance shown in \cref{fig21}.

\begin{figure}[tb]
\centering
\begin{subfigure}{0.50\textwidth}
\includegraphics[width=\linewidth]{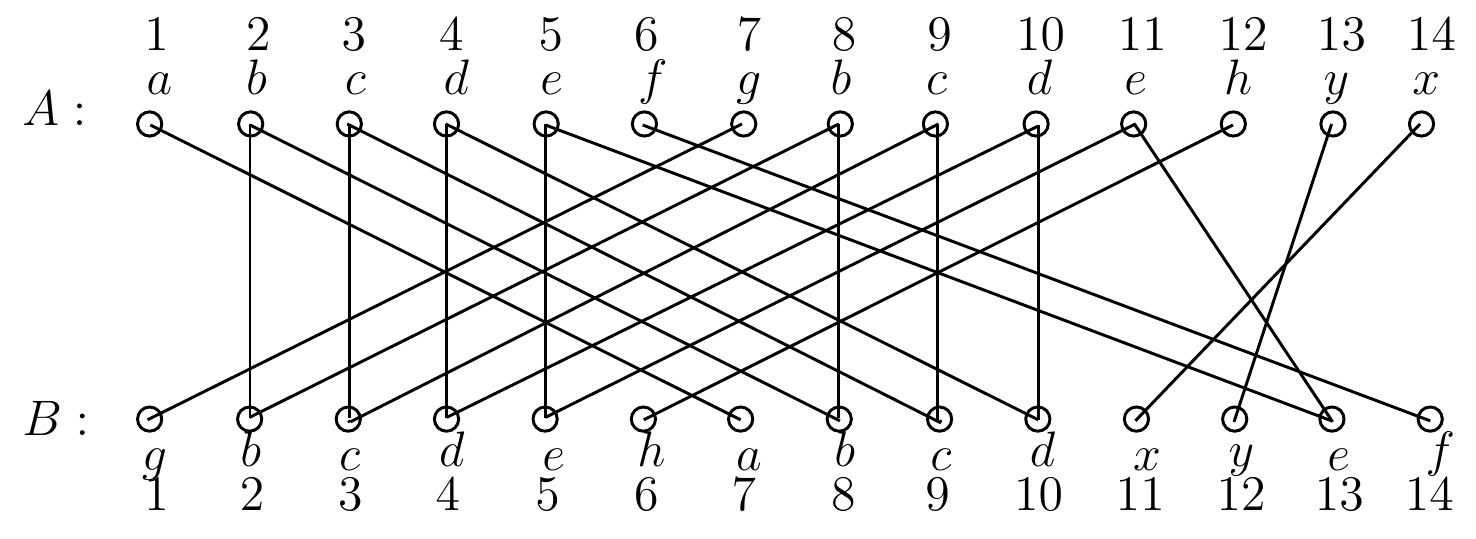}
\caption{The bipartite graph $H = (A, B, F)$.\label{fig24a}}
\end{subfigure}
\hspace*{\fill}
\begin{subfigure}{0.48\textwidth}
\includegraphics[width=\linewidth]{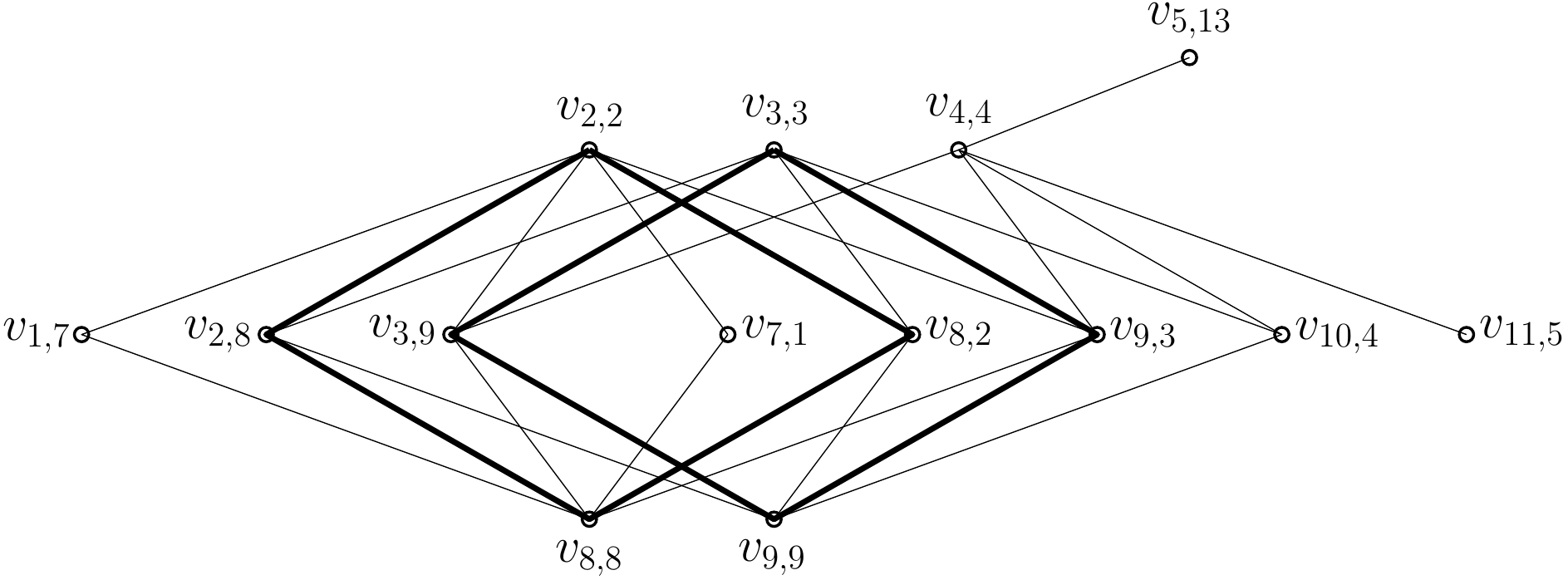}
\caption{The instance graph $G = (V, E)$.\label{fig24b}}
\end{subfigure}
	
\vskip 0.2in
\begin{subfigure}{0.42\textwidth}
\includegraphics[width=\linewidth]{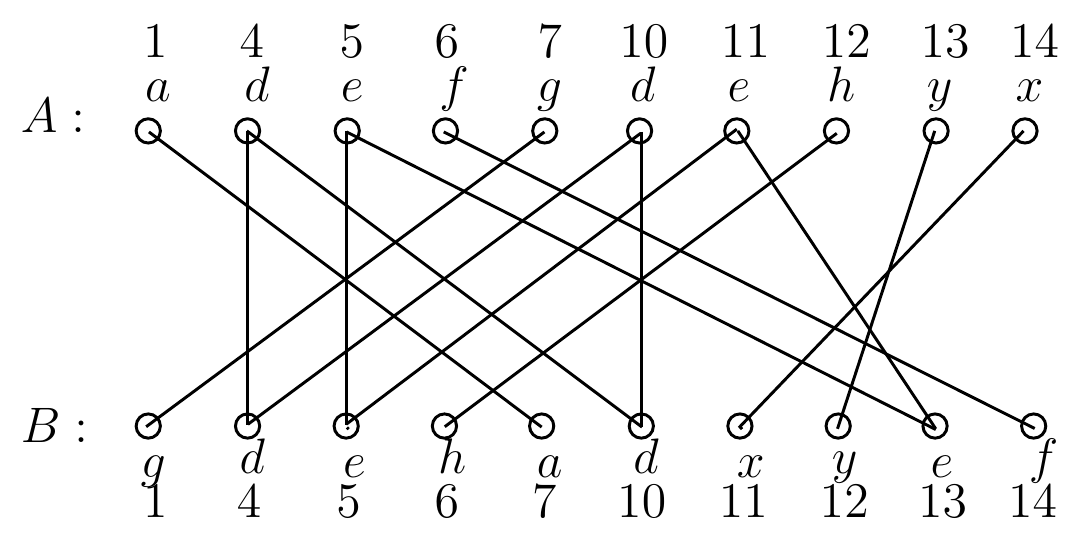}
\caption{The bipartite graph $H' = (A', B', F')$ after removal of ${\cal S}^2(2, 8; 2, 8)$.\label{fig24c}}
\end{subfigure}
\hspace*{\fill}
\begin{subfigure}{0.48\textwidth}
\includegraphics[width=\linewidth]{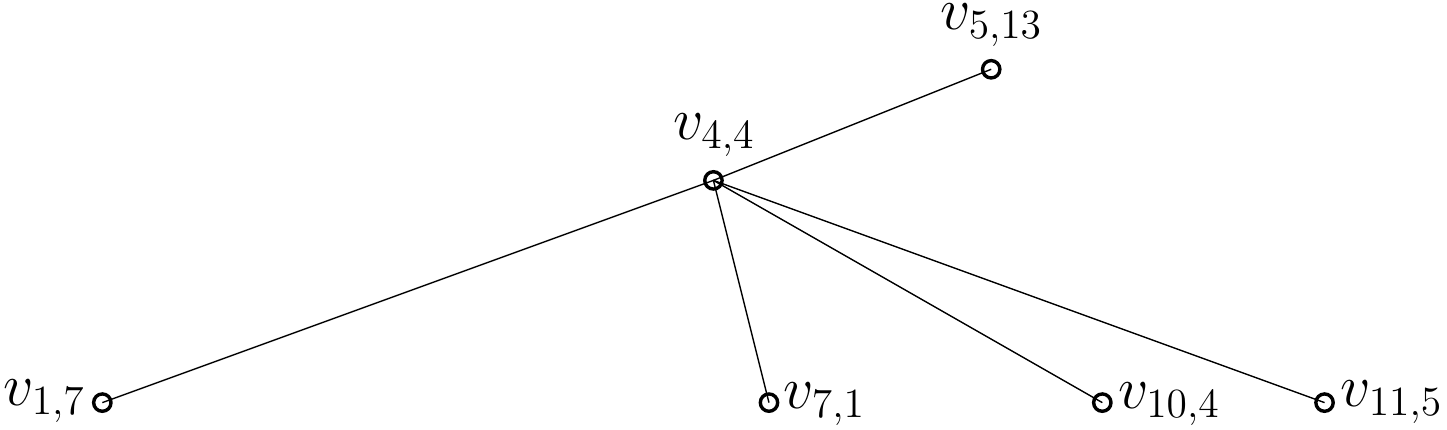}
\caption{The updated instance graph $G' = (V', E')$ after removal of ${\cal S}^2(2, 8; 2, 8)$.\label{fig24d}}
\end{subfigure}
\caption{An instance of the $2$-{\sc Max-Duo} problem with $A = (a,b,c,d,e,f,g,b,c,d,e,h,y,x)$ and $B = (g,b,c,d,e,h,a,b,c,d,x,y,e,f)$.
	The bipartite graph $H = (A, B, F)$ is shown in \cref{fig24a} and the instance graph $G = (V, E)$ of the {\sc MIS} problem is shown in \cref{fig24b}.
	There is a maximal series of $2$ squares ${\cal S}^2(2, 8; 2, 8)$ in the graph $G$, with the four substrings ``$bcd$''.
	The bipartite graph $H' = (A', B', F')$ is shown in \cref{fig24c} and the graph $G' = (V', E')$ is shown in \cref{fig24d},
	on $A' = (a,d,e,f,g,d,e,h,y,x)$ and $B' = (g,d,e,h,a,d,x,y,e,f)$.
	Applying the vertex contracting process on $G$ also gives the graph $G'$.\label{fig24}}
\end{figure}

\begin{lemma}
\label{lemma28}
Suppose ${\cal S}^p(i, i'; j, j')$, where $p \ge 1$, exists in the graph $G$.
Then,
\begin{enumerate}
\item
	the two substrings $(a_i, a_{i+1}, \ldots, a_{i+p})$ and $(a_{i'}, a_{i'+1}, \ldots, a_{i'+p})$ of the string $A$ and
	the two substrings $(b_j, b_{j+1}, \ldots, b_{j+p})$ and $(b_{j'}, b_{j'+1}, \ldots, b_{j'+p})$ of the string $B$ are identical and do not overlap;
\item
	if a maximum independent set of $G$ contains less than $2p$ vertices from ${\cal S}^p(i, i'; j, j')$,
	then it must contain either the four vertices $v_{i-1, j-1}, v_{i'-1, j'-1}, v_{i'+p, j+p}, v_{i+p, j'+p}$
	or the four vertices $v_{i'-1, j-1}, v_{i-1, j'-1}, v_{i+p, j+p}, v_{i'+p, j'+p}$.
\end{enumerate}
\end{lemma}
\begin{proof}
By the definition of the square $S(i+q, i'+q; j+q, j'+q)$, we have $a_{i+q} = a_{i'+q}$ and $a_{i+q+1} = a_{i'+q+1}$;
we thus conclude that the two substrings $(a_i, a_{i+1}, \ldots, a_{i+p})$ and $(a_{i'}, a_{i'+1}, \ldots, a_{i'+p})$ are identical.
In \cref{fig24b}, for ${\cal S}^2(2, 8; 2, 8)$ the two substrings are ``$bcd$''.
If these two substrings overlapped, then there would be three occurrences of at least one letter, contradicting the fact that $k = 2$.
This proves the first item.

Note that the square $S(i-1, i'-1; j-1, j'-1)$ does not exist in the graph $G$, and thus at most two of its four vertices
(which are $v_{i-1, j-1}, v_{i'-1, j-1}, v_{i-1, j'-1}$ and $v_{i'-1, j'-1}$) exist in $V$.
We claim that if no vertex of the square $S(i, i'; j, j')$ is in $I^*$,
then there are exactly two of the four vertices $v_{i-1, j-1}, v_{i'-1, j-1}, v_{i-1, j'-1}$ and $v_{i'-1, j'-1}$ exist in $V$ and they both are in $I^*$.
Suppose otherwise there is at most one of the four vertices in $I^*$, say $v_{i-1, j-1}$;
we may increase the size of $I^*$ by removing $v_{i-1, j-1}$ while adding
either the two vertices $v_{i, j}$ and $v_{i', j'}$ or the two vertices $v_{i', j}$ and $v_{i, j'}$
(depending on which vertices of the square $S(i+1, i'+1; j+1, j'+1)$ are in $I^*$),
a contradiction.

Assume next that a vertex of the square $S(i, i'; j, j')$ is in $I^*$, say $v_{i, j}$;
then due to maximality of $I^*$ and Lemma~\ref{lemma24} both $v_{i, j}$ and $v_{i', j'}$ are in $I^*$.
We claim and prove similarly as in the last paragraph that if no vertex of the square $S(i+1, i'+1; j+1, j'+1)$ is in $I^*$,
then there are exactly two of the four vertices $v_{i-1, j-1}, v_{i'-1, j-1}, v_{i-1, j'-1}$ and $v_{i'-1, j'-1}$ exist in $V$ and they both are in $I^*$.
If there is a vertex of the square $S(i+1, i'+1; j+1, j'+1)$ in $I^*$,
then it must be one of $v_{i+1, j+1}$ and $v_{i'+1, j'+1}$;
and due to maximality and Lemma~\ref{lemma24} both $v_{i+1, j+1}$ and $v_{i'+1, j'+1}$ are in $I^*$.
And so on;
repeatedly applying this argument, we claim and prove similarly that if no vertex of the square $S(i+p-1, i'+p-1; j+p-1, j'+p-1)$ is in $I^*$,
then there are exactly two of the four vertices $v_{i-1, j-1}, v_{i'-1, j-1}, v_{i-1, j'-1}$ and $v_{i'-1, j'-1}$ exist in $V$ and they both are in $I^*$.
If there is a vertex of the square $S(i+p-1, i'+p-1; j+p-1, j'+p-1)$ in $I^*$,
then it must be one of $v_{i+p-1, j+p-1}$ and $v_{i'+p-1, j'+p-1}$;
and due to maximality and Lemma~\ref{lemma24} both $v_{i+p-1, j+p-1}$ and $v_{i'+p-1, j'+p-1}$ are in $I^*$.

To summarize, we proved in the above two paragraphs that if $I^*$ contains less than $2p$ vertices from ${\cal S}^p(i, i'; j, j')$,
then there are exactly two of the four vertices $v_{i-1, j-1}, v_{i'-1, j-1}, v_{i-1, j'-1}$ and $v_{i'-1, j'-1}$ exist in $V$ and they both are in $I^*$;
and these two vertices are either $v_{i-1, j-1}$ and $v_{i'-1, j'-1}$ or $v_{i'-1, j-1}$ and $v_{i-1, j'-1}$.
Symmetrically, there are exactly two of the four vertices $v_{i+p, j+p}, v_{i'+p, j+p}, v_{i+p, j'+p}$ and $v_{i'+p, j'+p}$ exist in $V$ and they both are in $I^*$;
and these two vertices are either $v_{i+p, j+p}$ and $v_{i'+p, j'+p}$ or $v_{i'+p, j+p}$ and $v_{i+p, j'+p}$.
Clearly from the above, when the combination is $v_{i-1, j-1}$ and $v_{i'-1, j'-1}$ versus $v_{i+p, j+p}$ and $v_{i'+p, j'+p}$,
we may increase the size of $I^*$ to contain exactly $2p$ vertices from ${\cal S}^p(i, i'; j, j')$ without affecting any vertex outside of ${\cal S}^p(i, i'; j, j')$,
a contradiction.
Therefore, the only possible combinations are $v_{i-1, j-1}$ and $v_{i'-1, j'-1}$ versus $v_{i'+p, j+p}$ and $v_{i+p, j'+p}$,
and $v_{i'-1, j-1}$ and $v_{i-1, j'-1}$ versus $v_{i+p, j+p}$ and $v_{i'+p, j'+p}$.
This proves the second item of the lemma.
\end{proof}

Suppose ${\cal S}^p(i, i'; j, j')$, where $p \ge 1$, exists in the graph $G$.
Let $A'$ denote the string obtained from $A$ by removing the two substrings $(a_i, a_{i+1}, \ldots, a_{i+p-1})$ and $(a_{i'}, a_{i'+1}, \ldots, a_{i'+p-1})$
and concatenating the remainder together,
and $B'$ denote the string obtained from $B$ by removing the two substrings $(b_j, b_{j+1}, \ldots, b_{j+p-1})$ and $(b_{j'}, b_{j'+1}, \ldots, b_{j'+p-1})$ 
and concatenating the remainder.
Let the graph $G' = (V', E')$ denote the instance graph of the {\sc MIS} problem constructed from the two strings $A'$ and $B'$.
See for an example $G'$ in \cref{fig24d}, where there is a maximal series of $2$ consecutive squares ${\cal S}^2(2, 8; 2, 8)$ in the graph $G$.

\begin{corollary}
\label{coro29}
Suppose ${\cal S}^p(i, i'; j, j')$, where $p \ge 1$, exists in the graph $G$.
Then, the union of a maximum independent set in the graph $G' = (V', E')$ and certain $2p$ vertices from ${\cal S}^p(i, i'; j, j')$
becomes a maximum independent set in the graph $G = (V, E)$,
where these certain $2p$ vertices are $v_{i, j}, v_{i+1, j+1}, \ldots, v_{i+p-1, j+p-1}$ and $v_{i', j'}, v_{i'+1, j'+1}, \ldots, v_{i'+p-1, j'+p-1}$
if $v_{i-1, j-1}$ or $v_{i+p, j+p}$ is in the maximum independent set in $G'$,
or they are $v_{i', j}, v_{i'+1, j+1}, \ldots$, $v_{i'+p-1, j+p-1}$ and $v_{i, j'}, v_{i+1, j'+1}, \ldots, v_{i+p-1, j'+p-1}$
if $v_{i'-1, j-1}$ or $v_{i'+p, j+p}$ is in the maximum independent set in $G'$.
\end{corollary}
\begin{proof}
Consider the construction of the graph $G' = (V', E')$ from the two strings $A'$ and $B'$.
Equivalently, starting with the graph $G = (V, E)$,
if we contract the $p$ vertices $v_{i, j}, v_{i+1, j+1}, \ldots, v_{i+p-1, j+p-1}$ into the vertex $v_{i+p, j+p}$ if it exists or otherwise into a void vertex,
contract the $p$ vertices $v_{i', j'}, v_{i'+1, j'+1}, \ldots, v_{i'+p-1, j'+p-1}$ into the vertex $v_{i'+p, j'+p}$ if it exists or otherwise into a void vertex,
contract the $p$ vertices $v_{i', j}, v_{i'+1, j+1}, \ldots, v_{i'+p-1, j+p-1}$ into the vertex $v_{i'+p, j+p}$ if it exists or otherwise into a void vertex,
and contract the $p$ vertices $v_{i, j'}, v_{i+1, j'+1}, \ldots, v_{i+p-1, j'+p-1}$ into the vertex $v_{i+p, j'+p}$ if it exists or otherwise into a void vertex,
then we obtain a graph that is exactly $G'$.
In the graph $G'$, the vertices $v_{i-1, j-1}$ and $v_{i'+p, j+p}$, if both exist in $V$, become adjacent to each other;
so are the vertices $v_{i'-1, j-1}$ and $v_{i+p, j+p}$, if both exist in $V$.
It follows that the maximum independent set in the graph $G' = (V', E')$ does not contain both vertices $v_{i-1, j-1}$ and $v_{i'+p, j+p}$,
or both vertices $v_{i'-1, j-1}$ and $v_{i+p, j+p}$.
Therefore, starting with the maximum independent set in the graph $G' = (V', E')$,
we can add exactly $2p$ vertices from ${\cal S}^p(i, i'; j, j')$ to form an independent set in $G$, of which the maximality can be proved by a simple contradiction.

We remark that in the extreme case where none of the vertices of $S(i-1, i'-1; j-1, j'-1)$ and none of the vertices of $S(i+p, i'+p; j+p, j'+p)$
are in the maximum independent set in $G'$,
we may add either of the two sets of $2p$ vertices from ${\cal S}^p(i, i'; j, j')$ to form a maximum independent set in $G$.
\end{proof}

Iteratively applying the above string shrinkage process, or equivalently the vertex contracting process,
associated with the elimination of a maximal series of consecutive squares.
In $O(n)$ iterations, we achieve the final graph containing no squares, which we denote as $G_1 = (V_1, E_1)$.

\section{An approximation algorithm for $2$-{\sc Max-Duo}}
\label{sec3}
A high-level description of the approximation algorithm, denoted as {\sc Approx}, for the $2$-{\sc Max-Duo} problem is depicted in \cref{approx}.

\vskip -15pt
\begin{figure}[H]
\centering
\begin{algorithm}[H]
\footnotesize
\caption*{{\bf Algorithm \sc Approx}}
\begin{algorithmic}[1]
\State Construct the graph $G = (V, E)$ from two input strings $A$ and $B$;
\While{(there is a square in the graph)}
	\State find a maximal series of squares;
	\State locate the four identical substrings of $A$ and $B$ as in Lemma~\ref{lemma28};
	\State remove the corresponding substrings and accordingly update the graph;
\EndWhile
\State denote the resultant graph as $G_1 = (V_1, E_1)$;
\State set $L_1$ to contain all degree-$0$ and degree-$1$ vertices of $G_1$;
\State set $N[L_1]$ to be the closed neighborhood of $L_1$ in $G_1$, {\it i.e.} $N[L_1] = L_1 \cup N(L_1)$;
\State set $G_2 = G_1[V_1 - N[L_1]]$, the subgraph of $G_1$ induced on $V_1 - N[L_1]$;
\State compute an independent set $I_2$ in $G_2$ by the $((\Delta+3)/5 + \epsilon)$-approximation in \cite{BF99};
\State set $I_1 = I_2 \cup L_1$, an independent set in $G_1$; 
\State \Return an independent set $I$ in $G$ using $I_1$ and Corollary \ref{coro29}.
\end{algorithmic}
\end{algorithm}
\caption{A high-level description of the approximation algorithm for $2$-{\sc Max-Duo}.\label{approx}}
\end{figure}

In more details, given an instance of the $2$-{\sc Max-Duo} problem with two length-$n$ strings $A$ and $B$,
the first step of our algorithm is to construct the graph $G = (V, E)$, which is done in $O(n^2)$ time.
In the second step (Lines 2--7 in \cref{approx}),
it iteratively applies the vertex contracting process presented in \cref{sec2} at the existence of a maximal series of consecutive squares,
and at the end it achieves the final graph $G_1 = (V_1, E_1)$ which does not contain any square.
This second step can be done in $O(n^2)$ time too since each iteration of vertex contracting process is done in $O(n)$ time and there are $O(n)$ iterations.
In the third step (Lines 8--10 in \cref{approx}), let $L_1$ denote the set of singletons (degree-$0$ vertices) and leaves (degree-$1$ vertices) in the graph $G_1$;
our algorithm removes all the vertices of $L_1$ and their neighbors from the graph $G_1$ to obtain the remainder graph $G_2 = (V_2, E_2)$.
This step can be done in $O(n^2)$ time too due to $|V_1| \le |V| \le 2n$,
and the resultant graph $G_2$ has maximum degree $\Delta \le 4$ by Corollaries \ref{coro25} and \ref{coro26}.
(See for an example illustrated in \cref{fig31a}.)
In the fourth step (Lines 11--12 in \cref{approx}),
our algorithm calls the state-of-the-art approximation algorithm for the {\sc MIS} problem \cite{BF99}
on the graph $G_2$ to obtain an independent set $I_2$ in $G_2$;
and returns $I_1 = L_1 \cup I_2$ as an independent set in the graph $G_1$.
The running time of this step is dominated by the running time of the state-of-the-art approximation algorithm for the {\sc MIS} problem,
which is a high polynomial in $n$ and $1/\epsilon$.
In the last step (Line 13 in \cref{approx}), using the independent set $I_1$ in $G_1$,
our algorithm adds $2p$ vertices from each maximal series of $p$ consecutive squares according to Corollary \ref{coro29},
to produce an independent set $I$ in the graph $G$.
(For an illustrated example see \cref{fig31b}.)
The last step can be done in $O(n)$ time.

\begin{figure}[tb]
\centering
\begin{subfigure}{0.44\textwidth}
\includegraphics[width=\linewidth]{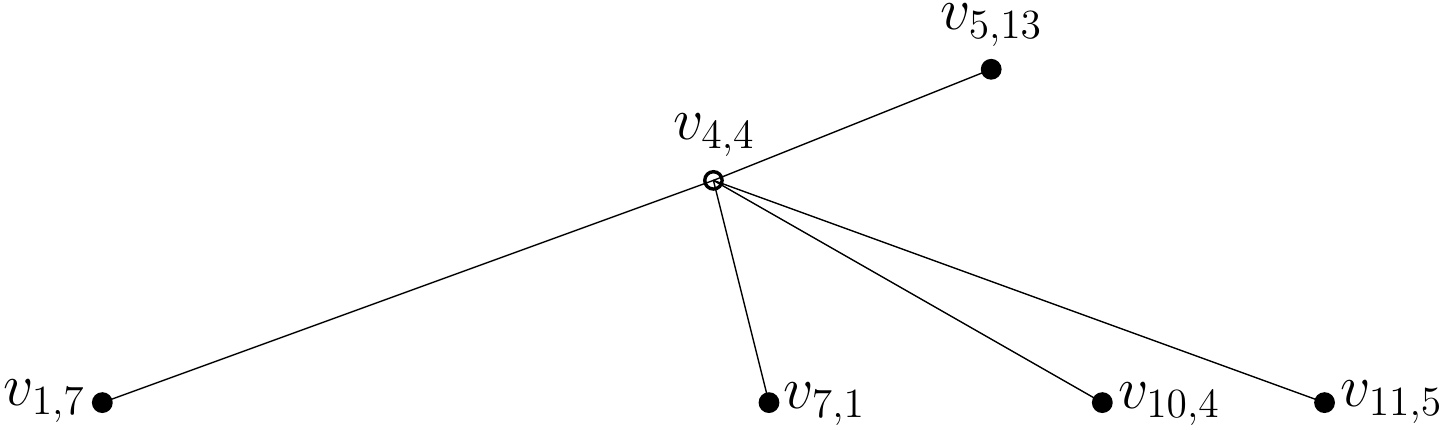}
\caption{The independent set $I_1 = \{v_{1,7}, v_{7,1}, v_{10,4}, v_{11,5}, v_{5,13}\}$ in $G_1$,
	consisting of all the five leaves of $G_1 = G'$ shown in \cref{fig24d}.\label{fig31a}}
\end{subfigure}
\hspace*{\fill}
\begin{subfigure}{0.515\textwidth}
\includegraphics[width=\linewidth]{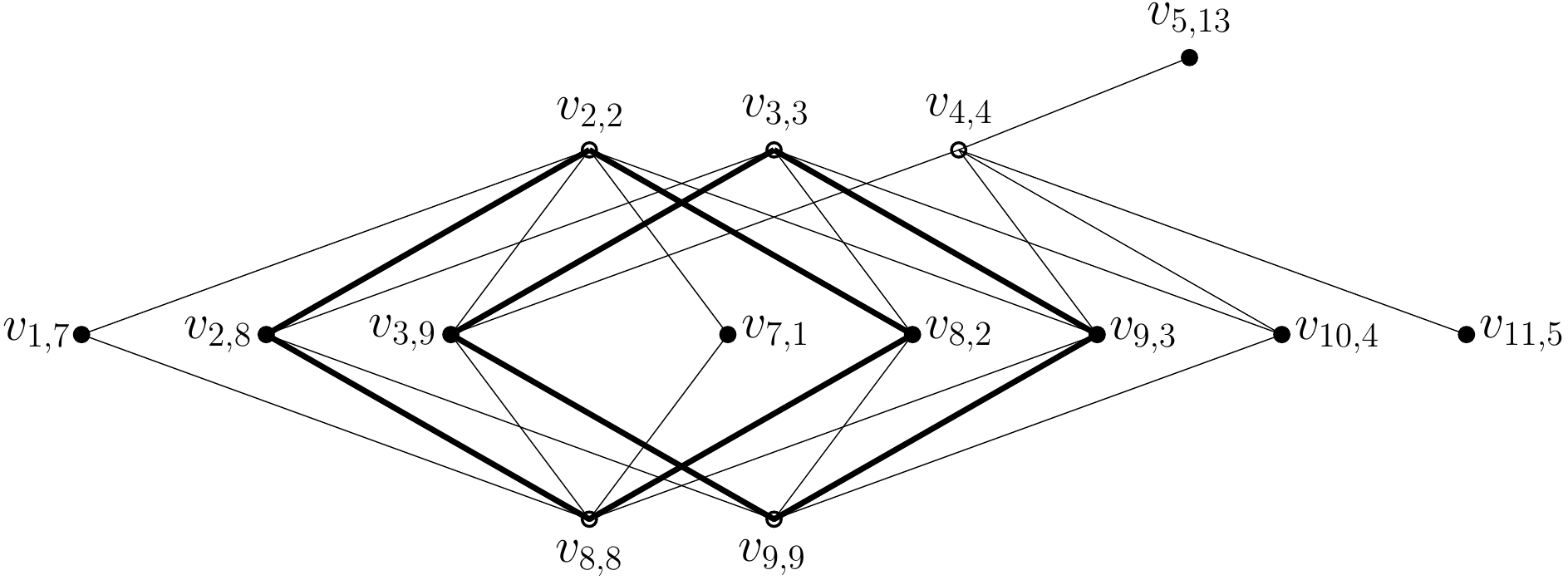}
\caption{Using $I_1$, since $v_{10, 4} \in I_1$,
	the four vertices $v_{2, 8}, v_{3, 9}$, $v_{8, 2}, v_{9, 3}$ are added
	to form an independent set $I$ in the original graph $G$ shown in \cref{fig24b}.\label{fig31b}}
\end{subfigure}

\vskip 0.2in
\begin{subfigure}{0.6\textwidth}
\includegraphics[width=\linewidth]{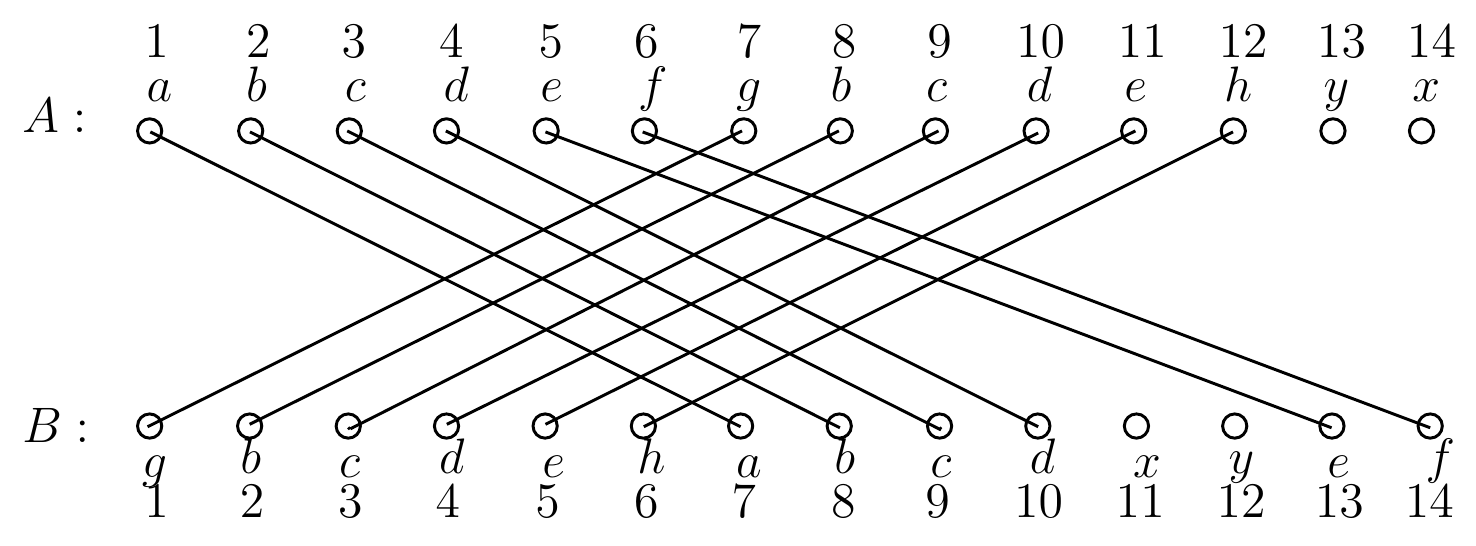}
\caption{The parallel edges of $H$ corresponding to the independent set $I$ shown in Figure \ref{fig31b},
	also correspond to the $9$ preserved duos $(a, b), (b, c), (c, d), (e, f)$, $(g, b), (b, c), (c, d), (d, e), (e, h)$
	for the instance shown in \cref{fig24a}.\label{fig31c}}
\end{subfigure}
\caption{Illustration of the execution of our algorithm {\sc Approx} on the instance shown in \cref{fig24}.
	The independent set $I_1$ in the graph $G_1$ is shown in \cref{fig31a} in filled circles,
	for which we did not apply the state-of-the-art approximation algorithm for the {\sc MIS} problem.
	The independent set $I$ in the graph $G$ is shown in \cref{fig31b} in filled circles,
	according to Corollary \ref{coro29} the four vertices $v_{2, 8}, v_{3, 9}$, $v_{8, 2}, v_{9, 3}$ are added due to $v_{10, 4} \in I_1$.
	The parallel edges of $H$ corresponding to the vertices of $I$ are shown in \cref{fig31c},
	representing a feasible solution to the $2$-{\sc Max-Duo} instance shown in \cref{fig24}.\label{fig31}}
\end{figure}

The state-of-the-art approximation algorithm for the {\sc MIS} problem on a graph with maximum degree $\Delta$
has a performance ratio of $(\Delta + 3)/5 + \epsilon$, for any $\epsilon > 0$ \cite{BF99}.

\begin{lemma}
\label{lemma31}
In the graph $G_1 = (V_1, E_1)$, let $\OPT_1$ denote the cardinality of a maximum independent set in $G_1$,
and let $\SOL_1$ denote the cardinality of the independent set $I_1$ returned by the algorithm {\sc Approx}.
Then, $\OPT_1 \le (1.4 + \epsilon) \SOL_1$, for any $\epsilon > 0$.
\end{lemma}
\begin{proof}
Let $L_1$ denote the set of singletons (degree-$0$ vertices) and leaves (degree-$1$ vertices) in the graph $G_1$;
our algorithm {\sc Approx} removes all the vertices of $L_1$ and their neighbors from the graph $G_1$ to obtain the remainder graph $G_2 = (V_2, E_2)$.
The graph $G_2$ has maximum degree $\Delta \le 4$ by Corollaries \ref{coro25} and \ref{coro26}.
Let $\OPT_2$ denote the cardinality of a maximum independent set in $G_2$,
and let $\SOL_2$ denote the cardinality of the independent set $I_2$ returned by the state-of-the-art approximation algorithm for the {\sc MIS} problem.
We have $\OPT_1 = |L_1| + \OPT_2$ and $\OPT_2 \le (1.4 + \epsilon) \SOL_2$, for any $\epsilon > 0$.
Therefore,
\[
\OPT_1 \le |L_1| + (1.4 + \epsilon) \SOL_2 \le (1.4 + \epsilon) (|L_1| + \SOL_2) = (1.4 + \epsilon) \SOL_1.
\]
This proves the lemma.
\end{proof}

\begin{theorem}
\label{thm32}
The $2$-{\sc Max-Duo} problem can be approximated within a ratio arbitrarily close to $1.4$, by a linear reduction to the {\sc MIS} problem.
\end{theorem}
\begin{proof}
We prove by induction.
At the presence of maximal series of $p$ consecutive squares, we perform the vertex contracting process iteratively.
In each iteration to handle one maximal series of $p$ consecutive squares,
let $G$ and $G'$ denote the graph before and after the contracting step, respectively.
Let $\OPT'$ denote the cardinality of a maximum independent set in $G'$,
and let $\SOL'$ denote the cardinality of the independent set $I'$ returned by the algorithm {\sc Approx}.
Given any $\epsilon > 0$, from Lemma~\ref{lemma31}, we may assume that $\OPT' \le (1.4 + \epsilon) \SOL'$.

Let $\OPT$ denote the cardinality of a maximum independent set in $G$,
and let $\SOL$ denote the cardinality of the independent set returned by the algorithm {\sc Approx},
which adds $2p$ vertices from the maximal series of $p$ consecutive squares to the independent set $I'$ in $G'$, according to Corollary \ref{coro29},
to produce an independent set $I$ in the graph $G$.
Lemma~\ref{lemma28} states that
$\OPT = \OPT' + 2p$.
Therefore,
\[
\OPT = \OPT' + 2p \le (1.4 + \epsilon) \SOL' + 2p \le (1.4 + \epsilon) (\SOL' + 2p) = (1.4 + \epsilon) \SOL.
\]
This proves that for the original graph $G = (V, E)$ we also have $\OPT \le (1.4 + \epsilon) \SOL$ accordingly.
That is, the worst-case performance ratio of our algorithm {\sc Approx} is $1.4 + \epsilon$, for any $\epsilon > 0$.
The time complexity of the algorithm {\sc Approx} has been determined to be polynomial at the beginning of the section,
and it is dominated by the time complexity of the state-of-the-art approximation algorithm for the {\sc MIS} problem.
The theorem is thus proved.
\end{proof}

\section{Conclusion}
\label{sec5}
In this paper, we examined the {\sc Max-Duo} problem, the complement of the well studied {\em minimum common string partition} problem.
Based on an existing linear reduction to the {\em maximum independent set} (MIS) problem \cite{GKZ04,BKL14},
we presented a vertex-degree reduction technique for the $2$-{\sc Max-Duo} to reduce the maximum degree of the constructed instance graph to $4$.
Along the way, we uncovered many interesting structural properties of the constructed instance graph.
This degree reduction enables us to adopt the state-of-the-art approximation algorithm for the {\sc MIS} problem on low degree graphs~\cite{BF99}
to achieve a $(1.4 + \epsilon)$-approximation for $2$-{\sc Max-Duo}, for any $\epsilon > 0$.

It is worth mentioning that our vertex-degree reduction technique can be applied for $k$-{\sc Max-Duo} with $k \ge 3$.
In fact, we had worked out the details for $k = 3$, to reduce the maximum degree of the constructed instance graph from $12$ to $10$,
leading to a $(2.6 + \epsilon)$-approximation for $3$-{\sc Max-Duo}, for any $\epsilon > 0$.
Nevertheless, the $(2.6 + \epsilon)$-approximation is superseded by the $(2 + \epsilon)$-approximation for the general {\sc Max-Duo}~\cite{DGO17}.

It would be worthwhile to investigate whether the maximum degree can be further reduced to $3$,
by examining the structural properties associated with the degree-$4$ vertices.
On the other hand, it is also interesting to examine whether a better-than-$1.4$ approximation algorithm can be designed directly for the {\sc MIS} problem
on those degree-$4$ graphs obtained at the end of the vertex contracting process.

\subparagraph*{Acknowledgements.}
All authors are supported by NSERC Canada.
Additionally, 
Chen is supported by the NSFC Grants No. 11401149, 11571252 and 11571087, and the China Scholarship Council Grant No. 201508330054;
Liu is supported by the NSFC Grant Nos. 61370052 and 61370156.
Luo is supported by the NSFC Grant No. 71371129 and the PSF China Grant No. 2016M592680;
Lin and Zhang are supported by the NSFC Grant No. 61672323.





\end{document}